\documentclass{article}
\usepackage{fullpage}
\usepackage{amsmath}
\usepackage{amssymb}
\usepackage{amsthm}
\usepackage{algorithmic}
\usepackage{algorithm}
\usepackage{url}

\newcommand{\myket}[1]{|#1\rangle}
\newcommand{\mybra}[1]{\langle#1|}
\newcommand{\Tr}{\text{Tr}}
\newcommand{\Fid}{\mbox{\textit{Fid}}}

\newtheorem{thm}{Theorem}
\newtheorem{cor}{Corollary}
\newtheorem{lem}{Lemma}
\newtheorem{defin}{Definition}
\newtheorem{exe}{Example}
\newtheorem{sce}{Scenario}
\newtheorem{rem}{Remark}
\newtheorem{que}{Open question}
\newtheorem{red}{Reduction}
\newtheorem{prop}{Proposition}
\newtheorem{conj}{Conjecture}
\newenvironment{qlt}[1][Quantum learning task]{\begin{trivlist}
\item[\hskip \labelsep {\bfseries #1}]}{\end{trivlist}}

\DeclareMathOperator*{\argmax}{arg\,max}

\title{Quantum classification}

\author{
S\'{e}bastien Gambs\\
D\'{e}partement d'informatique et de recherche op\'{e}rationnelle\\
Universit\'{e} de Montr\'{e}al\\
C.P. 6128, Succ.\ Centre-Ville, Montr\'{e}al (Qu\'ebec), H3C 3J7 \textsc{Canada}\\
\texttt{gambsseb\textnormal{@}iro.umontreal.ca} \\
}

\begin{document}

\maketitle

\begin{abstract}
Quantum classification is defined as the task of predicting the
associated class of an unknown quantum state drawn from an ensemble
of pure states given a finite number of copies of this state. By
recasting the state discrimination problem within the framework of
Machine Learning (ML), we can use the notion of learning reduction
coming from classical ML to solve different variants of the
classification task, such as the weighted binary and the multiclass
versions.
\end{abstract}

\section{Introduction}
\label{intro}

Suppose that you are given an unknown quantum state drawn from an
ensemble of possible pure states where each state is labeled after
the class from which it originated. How well can you predict the
class of this unknown state? This general question is often referred
to in the literature as (quantum) \emph{state
discrimination}\footnote{Other common names include \emph{state
distinguishability} and \emph{state
identification}.}~\cite{surveystatediscrimination} and has been
studied at least as far back as the seminal work of Helstrom in the
seventies in the field of \emph{quantum detection and estimation
theory}~\cite{helstrom:detection76}. Of course, the answer will
depend on parameters such as the structure and your knowledge of the
ensemble of pure states, the dimension of the Hilbert space in which
the quantum states live and the number of copies of the unknown
state you received.

In this paper, we take a Machine Learning (ML) view of the problem
by recasting it as a learning task called \emph{quantum
classification}. Our main goal by doing so is to bring new ideas and
insights from ML to help solve this task and some of its variants.
Other motivations include the characterization of these learning
tasks in terms of the amount of information needed to complete them
(measured for instance by the number of copies of the quantum
states) and the development of a framework that can be used to relate
and compare these tasks.

This approach of performing learning on quantum states was
originally taken and defined in~\cite{learningquantumworld}, where
it was illustrated by giving an explicit algorithm for the task of
\emph{quantum clustering}, where the goal is to group in clusters
quantum states that are similar (using the fidelity as a similarity
measure) while putting states that are dissimilar in different
clusters. The model of learning on quantum states put forward in
this paper is complementary to a model proposed by
Aaronson~\cite{learningquantumstates}, where the training dataset is
composed of POVM's (\emph{Positive-Operator Valued Measurement}),
and not quantum states. In Aaronson's model, we receive a finite
number of copies of an unknown quantum state and the goal is, by
``training'' this state on a few POVM's, to produce with high
probability a hypothesis that can generalize with a reasonable
accuracy on unobserved POVM's belonging to this training dataset.

The outline of this paper is as follows. First, the model of
performing learning in a quantum world is introduced in
Section~\ref{learning_quantum} along with the notion of learning
reduction which allows us to relate together different learning
tasks. Afterwards, in Section~\ref{binary_classification}, the task
of binary classification is described, and the weighted and
multiclass versions of this task are defined respectively in
Sections~\ref{weighted_binary_classification} and~\ref{multiclass
classification}. Finally, Section~\ref{discussion_and_conclusion}
concludes with a discussion.

\section{Learning in a quantum world}
\label{learning_quantum}

\emph{Machine Learning}
(ML)~\cite{book_pattern_classification,book_machine_learning,book_statistical_learning}
is the field that studies techniques to
\emph{give to machines the ability to learn from past experience}.
Typical tasks in \emph{supervised learning} include the ability to
predict the class (\emph{classification}) or some unobserved
characteristic (\emph{regression}) of an object based on some
observations. In \emph{unsupervised learning}, the goal is to find
some structure hidden within the data such as discovering
``natural'' clusters (\emph{clustering}), finding a meaningful low-dimensional representation of the data (\emph{dimensionality
reduction}) or learning explicitly a probability function (also
called density function) that represents the true distribution of
the data (\emph{density estimation}). ML algorithms learn from a
training dataset which contains observations about objects, which
are either obtained empirically or acquired from experts.

\subsection{Learning with a classical dataset}

In \emph{classical} ML, the observations and the objects are
implicitly considered to be classical and the machine which performs
the learning is assumed to be a classical computer (such as a
classical Turing machine or a classical logical circuit). For
instance, in \emph{supervised learning}, a training dataset
containing $n$ data points can be described as \mbox{$D_n =
\{(x_1,y_1),\ldots,(x_n,y_n)\}$}, where $x_i$ would be a vector of
\emph{observations} on the \emph{characteristics} of the
$i^{\mbox{\scriptsize th}}$ object (or data point) and $y_i$ is the
corresponding \emph{class} of that object. As a typical example,
each object can be described using $d$ real-valued attributes (i.e.
$x_i \in \mathbb{R}^d$) and if we are dealing with binary
classification (i.e. $y_i \in \{-1,+1\}$).
\begin{exe}[Classical classification tasks]
Recognition of the digital fingerprints or the face of a person (in
this case each class corresponds to a person), automatically
classify a news article as belonging to the ``culture'' or
``sports'' section, detection of frauds, music genre classification,
etc$\ldots$
\end{exe}
The main difference between supervised and \emph{unsupervised}
learning is that in the latter case, the $y_i$ values are unknown.
This could mean that we know the possible labels in general but not
the specific label of each data point, or that even the number of
classes and their labels are unknown to us.

\subsection{Learning with a quantum dataset}

In a quantum world, an ML algorithm still needs a training dataset
from which to perform learning, but this dataset now contains
\emph{quantum} objects instead of classical observations on
classical objects (the machine is also a quantum computer).
\begin{defin}[Quantum training dataset]
A \emph{quantum training dataset} containing $n$ pure quantum states
can be described as
$D_n=\{(\myket{\psi_1},y_1),\ldots,(\myket{\psi_n},y_n)\}$, where
$\myket{\psi_i}$ is the $i^{\mbox{\scriptsize th}}$ quantum state of
the training dataset and $y_i$ is the class associated with this
state.
\end{defin}
\begin{exe}[Quantum training dataset composed of pure states defined
on $d$ qubits] In the context where all the pure states in the
training dataset live in a Hilbert space formed by $d$ qubits and we
are interested in the task of binary classification ;
\smash{$\myket{\psi_i} \in \mathbb{C}^{2^d}$} and $y_i \in
\{-1,+1\}$.
\end{exe}
In this work, we will restrict ourselves to the case where the
states are quantum but the classes remain classical. Further
generalization is to consider the situation in which objects can be
in a quantum superposition of classes\footnote{Note that being in a
quantum superposition of classes is not equivalent to the classical
notion of data point belonging to several classes in a fuzzy or
probabilistic manner.}. Another extension of the model is to allow
the quantum states to be mixed, and not only pure.

\subsection{Learning classes}

One of the intrinsic difficulty of defining learning in the quantum
world comes from the many ways in which quantum states can be
specified in the training dataset. For instance, the training
dataset could contain a finite number of copies of each quantum
state or consist of a classical description of these (like an
explicit description of their density matrices). This latter case is
the most ``powerful'' in the sense of information theory because
from a classical description of a state, it is always possible in
principle to produce as many quantum copies as desired.

To formalize this notion, the concept of learning classes that
differ in the form of the training dataset, the learner's
technological sophistication and his learning goal was introduced in
\cite{learningquantumworld}.
\begin{defin}[Learning class]
For \emph{learning class} $\mathsf{L}_\text{\it goal}^\text{\it
context}$, let subscript ``\textit{goal}'' refer to the learning
goal and superscript ``\textit{context}'' to the form of the
training dataset and/or the technology to which the learner has
access.\end{defin} Possible values for \textit{goal} are $cl$, which
stands for doing ML with a classical purpose, in mind, and $qu$ for
ML with a quantum motivation. Similarly, the superscript
\textit{context} can be $cl$ for ``classical'' if everything is
classical (with a possible exception for the goal) or $qu$ if
something ``quantum'' is going~on in the learning process. Other
values for \textit{context} can be used when we need to be more
specific. For example, $\mathsf{L}_{cl}^{cl}$ corresponds to ML in
the usual sense, in which we want to use classical means to learn
from classical observations about classical objects. Another example
is $\mathsf{L}_{cl}^{qu}$, in which we have access to a quantum
computer to facilitate the learning but the goal remains to perform
a classical task: the quantum computer could serve to speed up the
learning process.

In this paper, we are only concerned with the specific case where
``\mbox{$\text{\it goal}={qu}$}\,''. \begin{defin}[Quantum learning
from the classical description of the quantum states]
$\mathsf{L}_{qu}^{cl}$ is defined as the \emph{learning class in
which we receive the classical descriptions of the quantum states}
from the training dataset (i.e.~\mbox{$D_n
=\{(\psi_1,y_1),\ldots,(\psi_n,y_n)\}$}, where $\psi_i$ is the
classical description of quantum state~$\myket{\psi_i}$).
\end{defin}
Learning becomes more challenging\footnote{Remark however that the classical description of a state is generally exponentially longer to write if it is represented classically as a string of bits compare to the corresponding quantum state in the form of qubits. Therefore, we can imagine a paradoxical situation where to describe classically the $2^{1000}$ amplitudes of quantum state defined on $1000$ qubits, we would need more memory than there are atoms in the universe, and this even if each atom could be used individually as a classical unit of memory (i.e. a bit). By contrast, if we can coherently manipulate the atoms and maintain them in superposition, $1000$ atoms would suffice to store the same state.}
when the dataset is available only
in its quantum form, in which case more copies make life easier as
we can potentially extract more information on the state. For
instance, a corollary of the \emph{Holevo
bound}~\cite{holevo:capacity73,corollaryHolevo} states that it is
impossible to extract more than $d$ classical bits of information
from a quantum state living in a Hilbert space formed by $d$ qubits.
Moreover, the \emph{no-cloning theorem}~\cite{nocloningtheorem}
forbids us to produce two identical copies of an unknown quantum
state. Finally, some \emph{tradeoffs exist between the amount of
information} that we can learn on a quantum state and \emph{the
corresponding perturbation} than this process will generate
(see~\cite{tradeoffinformationdisturbance} for instance).
\begin{defin}[Quantum learning from a finite number of copies of the
quantum states] $\mathsf{L}_{qu}^{\otimes s}$ is defined as the
\emph{learning class in which we are given at least $s$ copies of
each quantum state of the training dataset} (i.e.~\smash{$D_n
=\{(\myket{\psi_1}^{\otimes s},y_1),\ldots,(\myket{\psi_n}^{\otimes
s},y_n)\}$}; where \smash{$\myket{\psi_i}^{\otimes s}$} symbolizes
$s$ copies of state $\myket{\psi_i}$).
\end{defin}
Contrast these classes with ML in a classical world (such as
$\mathsf{L}_{cl}^{cl}$), in which additional copies of a particular
object are obviously useless as they do not carry new information.
The main purpose of defining quantum learning classes is to be able
to put some quantum training datasets and some learning tasks within
them.

The quantum learning classes form a \emph{hierarchy in an
information-theoretic sense}, where the higher a class is located
inside the hierarchy, the more information it contains in order to
realize tasks linked to the datasets belonging to this class. The
class $\mathsf{L}_{qu}^{cl}$ is at the top of the hierarchy since it
corresponds to having a \emph{complete knowledge} about the quantum
states forming the training set. Let $\equiv_{\ell}$, $\leq_{\ell}$
and $<_{\ell}$ be the operators which denote respectively the
\emph{equivalence}, the \emph{weaker or equal} and the
\emph{strictly weaker} relationships within the hierarchy. The
following propositions (first stated in~\cite{learningquantumworld})
describe some relations between the learning classes forming the
hierarchy.
\begin{prop}\label{prop_1}
$\mathsf{L}_{qu}^{\otimes s} \equiv_{\ell} \mathsf{L}_{qu}^{cl}$ as
$s \rightarrow \infty$.
\begin{proof}
When the number of copies tends to infinity, it is always possible
to estimate $\myket{\psi}$ using \emph{quantum tomography} and
reconstruct the classical description with arbitrary precision.
\end{proof}
\end{prop}
\begin{prop}
$\mathsf{L}_{qu}^{\otimes 1} \leq_{\ell} \ldots \leq_{\ell}
\mathsf{L}_{qu}^{\otimes s} \leq_{\ell} \mathsf{L}_{qu}^{\otimes
s+1} \leq_{\ell} \ldots \leq_{\ell} \mathsf{L}_{qu}^{cl}$.
\begin{proof}
Each new copy of a state gives potentially more information on that
state. Therefore for any positive integer $s$, we have
\mbox{\smash{$\mathsf{L}_{qu}^{\otimes s} \leq_{\ell}
\mathsf{L}_{qu}^{\otimes s+1}$}}, which implies that if a learning
task $A\in\mathsf{L}_{qu}^{\otimes s}$, it also belongs to
$\mathsf{L}_{qu}^{\otimes s+1}$. Furthermore due to
Proposition~\ref{prop_1}, a classical description of a state is as
good as any number of copies.
\end{proof}
\end{prop}
\begin{prop}
$\mathsf{L}_{qu}^{\otimes s} + \mathsf{L}_{qu}^{\otimes 1}
\leq_{\ell} \mathsf{L}_{qu}^{\otimes s+1}$, where ``$+$'' denotes a
restriction that the first $s$ copies must be measured separately
from the the last.
\begin{proof}
Performing a \emph{joint} measurement by allowing $s+1$ copies to
interact together can potentially give more information than
performing a joint measurement on $s$ copies plus a separated
measurement on another copy. (See~\cite{jointmeasurement:91} for a
specific instance where $s=1$ and~\cite{covariantmeasurements:04}
for results about arbitrary $s$.)
\end{proof}
\end{prop}
An interesting open question is whether or not this hierarchy is
strict.
\begin{que}[Strict hierarchy of learning classes]
In the expression $\mathsf{L}_{qu}^{\otimes 1} \leq_{\ell} \ldots
\leq_{\ell} \mathsf{L}_{qu}^{\otimes s} \leq_{\ell}
\mathsf{L}_{qu}^{\otimes s+1} \leq_{\ell} \ldots \leq_{\ell}
\mathsf{L}_{qu}^{cl}$, can some of these $\leq_{\ell}$ be replaced
by~$<_{\ell}$\,?
\end{que}
There are good reasons to believe that the answer is positive since
it is usually the case that more information can be obtained about a
quantum state when more copies are available. Moreover, it has been
proven that in some situations that joint measurements are more
informative than individual
measurements~\cite{jointmeasurement:91,covariantmeasurements:04}.
However, it does not \emph{necessarily} follow that this additional
information can be used in a constructive manner to solve some
learning tasks.

\subsection{Learning reduction}

The notion of \emph{reduction between learning
tasks}~\cite{learningreductions} was developed and formalized during
these last years in the context of classical ML by Langford and
co-authors\footnote{See for instance the webpage of Langford's
project on learning reductions
\url{http://hunch.net/~jl/projects/reductions/reductions.html}.}.
\begin{defin}[Learning reduction~\cite{learningreductions}]
A \emph{learning task \emph{A} reduces to some other learning task
\emph{B}} if by having access to a black-box (an oracle) that solves
\emph{B}, it is also possible to solve \emph{A}.
\end{defin}
A learning reduction can be seen as an \emph{information-theoretic
statement} about how well it is possible to solve a particular
learning task given an algorithm (modelled abstractly by an oracle)
that can solve another task. Although in general it is desirable for
this transformation to be efficient, learning reductions differ from
the ``traditional'' reductions used in complexity theory (such as
Turing or Karp reductions) in the sense they do not try to
characterize the computational time needed to solve a particular
task. Rather, learning reductions offer a way to compare and relate
two different learning tasks in the sense of information-theory. If
\emph{A} reduces to \emph{B}, it means that any progress made on how
to solve \emph{B} can be transferred directly to \emph{A} by using
the reduction. Moreover, if different tasks all reduce to a single
\emph{learning primitive}, it means that an improvement on this
primitive has a direct impact on all the other tasks. For instance,
in sections~\ref{weighted_binary_classification} and~\ref{multiclass
classification}, we will see how to solve the \emph{weighted binary}
and the \emph{multiclass classification} tasks given an oracle for
solving the \emph{standard binary classification}
(section~\ref{binary_classification}).

A good reduction often offers some guarantee on how well the
performance of the black-box in solving problem \emph{B} also
implies a good performance regarding problem \emph{A}. For instance
in classification, this guarantee could take the form of
\emph{upper bounds on the error achieved by the final classifier}. The
upper bounds generally relate the average error of the classifiers
generated by the oracle on subproblems \emph{B} to the global error
that the final combined classifier will make on the general problem
\emph{A}.

\begin{defin}[(Training) error]
The \emph{(training) error} $\epsilon$ (or \emph{error rate}) of a
classifier $f$ is defined as the probability that this classifier
predicts the wrong class $y_i$ on a quantum state $\myket{\psi_i}$
drawn randomly from the states of the quantum training dataset
$D_n$. Formally:
\begin{equation}
\epsilon_f = \frac{1}{n}
\sum_{i=1}^{n}\text{\emph{Prob}}(f(\myket{\psi_i})\neq y_i)
\end{equation}
\end{defin}

This definition characterizes precisely the training error of the
classifier but not its \emph{generalization error}, which
corresponds to how well the classifier predicts on states that it
has not observed exactly beforehand (i.e. that are not part of the
training dataset). For now, we will focus only on the minimization
of this training error but we will come back to the generalization
error (which is really the essence of ML) in the discussion
(Section~\ref{discussion_and_conclusion}).

In the context of quantum classification, the notion of
\emph{regret} also takes a particular importance.
\begin{defin}[Regret]
The \emph{regret} $r$ of a classifier $f$ is defined as the
difference between its error rate $\epsilon_f$ and the smallest
achievable error $\epsilon_{opt}$ that can be achieved on the same
problem. Formally:
\begin{equation}
r_f = \epsilon_f - \epsilon_{opt}
\end{equation}
\end{defin}
The regret of a classifier, as well as its error, can potentially
take any value in the range between zero and one. The concept of
regret is particularly meaningful in the context of hard learning
problems, where the raw error rate alone is not an appropriate
measure to characterize the inherent difficulty of the learning.
Indeed, in some learning situations, it is possible to observe a
high error rate but a low (or even null) regret. In the classical
setting, a high error rate but a low regret is an indication of a
high level of noise. The situation is different in the quantum world
where a high error rate might be due to the intrinsic physical
difficulty of distinguishing two classes, but does not necessarily
imply a high level of noise. Regardless of the context, if the
regret of a classifier is zero, it essentially means that this
classifier is optimal.

Quantumly, a reduction or a learning task may also have a
\emph{cost} associated with it. Indeed, each call to the oracle may
require sacrificing some copies of the quantum states due to the
measurements performed by the oracle during the training. This cost
is measured in terms of the number of copies required
\emph{individually} for each quantum state of the training dataset.
Another way to define this cost would have been to count
\emph{globally} the number of copies required relative to the size
of the training dataset\footnote{Which is generally the same as
multiplying the individual cost by a factor linear in $n$, the
number of states in the training dataset.}. The cost can be
differentiated between the number of copies needed during the
\emph{training/learning phase}, where we learn/build a POVM $f$ that
acts as the classifier, and during the \emph{classification time}
(or \emph{testing phase}) where we use $f$ to classify an unknown
quantum state $\myket{\psi_?}$.

\begin{defin}[Training/learning cost]
The \emph{training/learning cost} of a reduction is equal to the
number of calls to the oracle made by the reduction, multiplied by
the number of copies of each quantum states that are used in each
call. In the case of a learning task, the cost is directly
caracterised by the number of copies of each state necessary to
perform this task.
\end{defin}

If we have a classical description of the quantum states (i.e. $D_n
\in \mathsf{L}_{qu}^{cl}$), the training/learning does not cost anything in
terms of information because we already have complete knowledge of
the quantum states.

\begin{defin}[Classification cost]
The \emph{classification cost} corresponds to the number of copies
of the unknown quantum state $\myket{\psi_?}$ that will be used by
the classifier to predict the class $y_?$ of this state.
\end{defin}

In the next three sections, we will define respectively the quantum
analogues of three learning tasks: \emph{binary classification},
\emph{weighted binary classification} and \emph{multiclass
classification}.

\section{Binary classification}\label{binary_classification}

The task of \emph{binary classification} consists in predicting the
class $y_?\in\{-1,+1\}$ of an unknown quantum state
$\myket{\psi_?}$, given a \emph{single copy of this
state}\footnote{See however the work of Sasaki and
Carlini~\cite{templatematching:02} for the case of more than one
copy of the unknown state $\myket{\psi_?}$ are available.}.
Formally, this learning task can be defined in the following manner.

\begin{qlt}[(Quantum) binary classification]
\label{binary_class}
:\\
\textbf{Input}: \mbox{$D_n
=\{(\myket{\psi_1},y_1),\ldots,(\myket{\psi_n},y_n)\}$}, a quantum training
dataset, where
$\myket{\psi_i} \in \mathbb{C}^{2^d}$ and $y_i \in \{-1,+1\}$.\\
\textbf{Output}: A POVM acting as a binary classifier $f$ that can
predict the class $y_?$ of an unknown quantum state
$\myket{\psi_?}$ given a single copy of this state.\\
\textbf{Goal}: Construct a binary classifier $f$ that minimizes the
training error $\epsilon_f = \frac{1}{n}\sum_{i=1}^{n}
\text{Prob}(f(\myket{\psi_i})\neq y_i)$.
\end{qlt}

A natural question to ask is what is the best probability of success
we can hope for, or, equivalently, the smallest error rate
achievable. The easiest situation occurs when we have \emph{complete
classical knowledge of the quantum states which compose the training
dataset} ($D_n \in \mathsf{L}_{qu}^{cl}$). However, even in this
case, it is not generally possible to devise a process that always
correctly classifies any unknown state from a single copy of this
state. This remains true even if we know in advance that this state
corresponds \emph{exactly} to one of the states in the training
set\footnote{Unless we are in the trivial situation where all the
states are mutually orthogonal. In this case, a
\emph{non-destructive measure} in a basis formed by these states
will reveal the state without perturbing it.}. From the classical
description of the states, it is possible to analytically build the
optimal POVM that minimizes the training error. Of course, it
remains to be seen how such an approach would \emph{generalize when
faced with a state which does not belong to the training set}. This
fundamental question will be briefly discussed in
Section~\ref{discussion_and_conclusion}.

Let $m_-$ be the number of quantum states in $D_n$ for which $y_i =
-1$ (negative class), and its complement $m_+$ be the number of
states for which \mbox{$y_i = +1$} (positive class), such that $m_-
+ m_+ = n$, the total number of data points in $D_n$. Moreover,
$p_-$ is the \emph{a priori} probability of observing the negative
class and is equal to $p_-=\frac{m_-}{n}$, and $p_+$ its
complementary probability for the positive class such that $p_- +
p_+ = 1$.

\begin{defin}[Statistical mixture of the negative class]
The
\emph{statistical mixture representing the negative class} $\rho_-$, is defined as
\smash{$\frac{1}{m_-} \sum_{i=1}^n
\text{I}\{y_i=-1\}\myket{\psi_i}\!\mybra{\psi_i}$},
where $\text{I}\{.\}$ is the \emph{indicator function} which equals 1
if its premise is true and 0 otherwise.
\end{defin}
\begin{defin}[Statistical mixture of the positive class]
In the same manner, the \emph{statistical mixture representing the positive class} $\rho_+$ is defined as \smash{$\frac{1}{m_+}
\sum_{i=1}^n \text{I}\{y_i=+1\}\myket{\psi_i}\!\mybra{\psi_i}$}.
\end{defin}
The problem of classifying an unknown state $\myket{\psi_?}$ drawn
from the training set is equivalent to distinguish between the mixed
states $\rho_-$ and~$\rho_+$. Consider for instance the following
scenario which illustrates this idea.
\begin{sce}[Preparation of the state of a class by a demon\footnote{This scenario could be reformulate by replacing the demon by a probabilistic algorithm. This raises the question of how much classical memory will the algorithm need to remember the description of the states.}]
Imagine a demon that sits inside a black-box with a single button.
Each time the button is pressed, the demon chooses at random between
the negative and positive class according to their \emph{a priori}
probabilities $p_-$ and $p_+$. Once the class is determined, the
demon chooses uniformly at random one of the states belonging to
this class and prepares the corresponding state (we suppose that the
demon in its infinite power knows the classical description of the
states and can prepare perfectly any one of them). This state is
returned as output by the black-box. Therefore, finding the class of
this state is essentially the same as guessing which class the
demon\footnote{Here the role of the demon is simply to prepare the
state, and not to act as an adversary which tries to fool the
learner who is outside the box.} has chosen during the first step,
but not necessarily identifying the exact state.
\end{sce}

The minimal error rate of this classification process is linked to the
\emph{statistical overlap} of the mixtures $\rho_-$ and $\rho_+$. This
kind of problem has already been studied in \emph{quantum detection
and estimation theory}~\cite{helstrom:detection76}, a field that
predates quantum information processing. Some results from
this field can be used to \emph{give bounds on the best training
error} that quantum learning algorithms might reach.

\begin{thm}[Helstrom
measurement~\cite{helstrom:detection76}]\label{Helstrom_thm} The
error rate of distinguishing between the two classes $\rho_-$ and
$\rho_+$ is bounded from below by $\epsilon_{hel} = \frac{1}{2} -
\frac{D(\rho_-,\rho_+)}{2}$, where $D(\rho_-,\rho_+) =
\text{\emph{Tr}} |p_-\rho_- - p_+\rho_+|$ is a distance measure
between $\rho_-$ and $\rho_+$ called the \emph{trace distance}
(here, $p_-$~and $p_+$ represent the \emph{a priori} probabilities
of classes $\rho_-$ and $\rho_+$, respectively). Moreover, this
bound can be achieved exactly by the optimal POVM called the
\emph{Helstrom measurement}.
\end{thm}
\begin{cor}[Regret of Helstrom measurement]
The Helstrom's measurement is a binary classifier that has a \emph{null regret}, which means $r_{hel}=0$.
\begin{proof}
The null regret of the Helstrom measurement follows directly from
the optimality of this POVM to distinguish between the two classes.
\end{proof}
\end{cor}

\begin{rem}[Error rate of the Helstrom measurement for equiprobable
classes]\label{Helstrom_rem} Consider the case where both the
negative class and the positive class are equiprobable. If $\rho_-$
and $\rho_+$ are two density matrices which correspond to the same
state, their trace distance $D(\rho_-,\rho_+)$ is equal to zero,
which means that the error $\epsilon_{hel}$ of the Helstrom
measurement is $\frac{1}{2}$. On the other hand, if $\rho_-$ and
$\rho_+$ are orthogonal, this means that $D(\rho_-,\rho_+)=1$ and
that the Helstrom measurement has an error $\epsilon_{hel}=0$.
\end{rem}
The purpose of a learning algorithm in the quantum setting is to
give a constructive way to come close to (or to achieve) the
Helstrom bound. If we know the classical description of the quantum
states, it corresponds to finding an efficient implementation of the
Helstrom measurement. If $D_n \in \mathsf{L}_{qu}^{\otimes s}$, the
learning becomes more challenging and it is difficult to
characterize the exact relationship between the number $s$ of copies
of each training state that are available, the dimension $d$ of the
Hilbert space in which the quantum states lives and the minimal
error rate $\epsilon$ we can hope to reach. Contrary to classical
ML, where it is always possible (but not recommended in terms of
generalization) to bring the training error down to zero (for
instance using a memory-based classifier such as 1-nearest
neighbour), the situation is different in the quantum context as
expressed by the following lemma.

\begin{lem}It is
\emph{impossible to reach a training error of zero in the quantum
case from a single copy of an unknown quantum state} unless of the
states of the training dataset are mutually orthogonal.
\begin{proof}
From Theorem~\ref{Helstrom_thm} and Remark~\ref{Helstrom_rem}, it is
easy to see that it is impossible to construct a POVM that perfectly
classifies a quantum state drawn from the training set $D_n$, unless
all the states of the ensemble are mutually orthogonal, or
equivalently that the distance between the two density matrices of
the classes is $D(\rho_-,\rho_+)=1$.
\end{proof}
\end{lem}
Given a finite number of copies of each state of the training set,
the possible \emph{learning strategies} include:
\begin{itemize}
\item[(1)] the estimation of the training set by making measurements (joint
or not) on some of the copies to construct a POVM that will
differentiate between the two classes,
\item[(2)] the design of a classification mechanism that uses the
copies only when the time of classifying an unknown quantum state
$\myket{\psi_?}$ comes or
\item[(3)] any hybrid strategy between (1) and (2).
\end{itemize}

For the classification, several \emph{measurement strategies} exist
in the quantum context such as:
\begin{itemize}
\item[(a)] \emph{maximizing the probability of predicting the class of an
unknown quantum state} (which corresponds to the Helstrom
measurement~\cite{helstrom:detection76}),
\item[(b)] \emph{minimizing the probability of making a wrong guess}. This
strategy is called \emph{unambiguous
discrimination}~\cite{unambiguousbounds:05} and is possible only
when the states of $D_n$ are linearly independent. In this specific
case, it is possible to design a measurement that is allowed to
sometimes answer ``I don't know'', but when it makes a prediction
regarding one the classes we can be $100\%$ confident than its
prediction is correct.
\item[(c)] \emph{any strategy between these two extremes} (a) and
(b). A \emph{confidence-based measurement}\footnote{The original
term is \emph{maximum-confidence
measurement}.}~\cite{confidencebasedmeasurements} is a measure that
can identify the class of a state with some confidence (that is
known), or answer ``I don't know'' the rest of the time. For a fixed
chosen confidence, the main objective when we build such a measure,
is to minimize the probability that it outputs ``I don't know''.
When the confidence is fixed at $100\%$ this directly corresponds to
the unambiguous discrimination, whereas if an inconclusive answer is
not allowed it corresponds to the Helstrom measurement. It is
sometimes possible to design a confidence-based measurement (with a
confidence greater than the Helstrom measurement) even when perfect
unambiguous discrimination is impossible (for instance if the states
of $D_n$ are linearly dependent).
\end{itemize}

In this paper, we will focus only (exception made of
section~\ref{class_state_identification}) on the measurement
strategy of maximizing the probability of identifying correctly the
class of a state (measurement strategy (a)) by learning from the
training dataset a POVM that can act as a classifier (learning
strategy (1)). We will make the assumption that we have access to an
oracle, called the \emph{Helstrom oracle}, than can efficiently
solve the task of binary classification.

\begin{defin}[Helstrom oracle]
The \emph{Helstrom oracle} is an abstract construction that takes as
input:
\begin{itemize}
\item[]\emph{Version 1:} a classical description of the density matrices $\rho_-$
and $\rho_+$ and their \emph{a priori} probabilities $p_-$ and $p_+$
(learning class $\mathsf{L}_{qu}^{cl}$) or
\item[]\emph{Version 2:} a finite number of copies of each state of the quantum
training dataset $D_n$ (learning class $\mathsf{L}_{qu}^{\otimes
\Theta(t_{bin})}$).
\end{itemize}
From this input, the oracle can be ``trained'' to produce an
efficient implementation (exact or approximative) of the POVM of the
Helstrom measurement $f_{hel}$, in the form of a quantum circuit
that can distinguish between $\rho_-$ and $\rho_+$. In the second
version of the oracle, its training cost $t_{bin}$ corresponds to
the minimum amount of copies of each state of the training dataset
that the oracle has to sacrifice in order to construct $f_{hel}$.
\end{defin}

One fundamental question deals with the (non-)existence of an
efficient implementation for the Helstrom measurement.

\begin{que}[Efficient implementation of the Helstrom measurement]
What are the learning situations (i.e. the ensembles of quantum
states) for which it is possible to implement efficiently (for
instance with a polynomial-size circuit) an approximate version of
the Helstrom measurement?
\end{que}

There is no \emph{a priori} guarantee that the description of the
POVM which corresponds to the Helstrom measurement can be physically
realized by a quantum circuit whose size is polynomial in the number
of input qubits. Indeed in the worst case, it could happen that this
circuit requires a number of gates that is exponential in its input
size, and this even for its approximate version.

By assuming the existence of the Helstrom oracle, we deliberately avoid the
burden of describing explicitly how the learning algorithm, which
acts as the oracle in practice, works (and how many quantum states
it requires for the learning process). Designing a learning
algorithm that can solve the binary classification task in practice
is a fundamental open question.

\begin{que}[Construction of a learning algorithm implementing the
Helstrom oracle] Is it possible to design a learning algorithm that
implements explicitly the Helstrom oracle? If so, what would be the
value of $t_{bin}$, the minimum number of copies of each training
state, that this algorithm requires during the learning?
\end{que}

This is a fundamental question on its own but instead we focus on
what tasks could be solved if we have access to such an oracle. If
we know a learning algorithm which has a low -- albeit not optimal
-- error rate, it is possible to use it instead of the Helstrom
oracle in almost all the reductions described in this paper.

Suppose that we have a binary classifier $f$ that can predict the
class of an unknown quantum state $\myket{\psi_?}$ with an error
$\epsilon$, for $\epsilon < \frac{1}{2}$. If we have access to a
constant number of copies of $\myket{\psi_?}$, we can simply repeat
the application of this classifier and output the majority of its
predictions. By standard Chernoff argument, this will diminishes the
probability of making an error exponentially fast with the number of
copies spent. This is true in the quantum world due to the inherent
probabilistic nature of the measurement process. In classical ML,
the situation is different as generally classifiers behave in a
deterministic manner, meaning that they will always predict the same
outcome when we present them with the same data point.

\section{Weighted Binary Classification} \label{weighted_binary_classification}

The \emph{weighted binary classification} task is similar to the
standard binary case, except that now each data point has a
\emph{weight} $w$ associated to it that indicates the importance of
correctly classifying this state. This weight can represent for
instance a penalty that we have to pay if we predict the wrong class
for this object. If $w=\frac{1}{n}$ for each state, then this corresponds to the standard binary classification.
\begin{qlt}[(Quantum) weighted binary classification]
\label{binary_class}
:\\
\textbf{Input}: $D_n
=\{(\myket{\psi_1},y_1,w_1),\ldots,(\myket{\psi_n},y_n,w_n)\}$, a
quantum training dataset, where
$\myket{\psi_i} \in \mathbb{C}^{2^d}$, $y_i \in \{-1,+1\}$ and $w_i \in [0,+\infty)$.\\
\textbf{Output}: A POVM acting as a binary classifier $f$ that can
predict the class $y_?$ of an unknown quantum state
$\myket{\psi_?}$.\\
\textbf{Goal}: Construct a binary classifier $f$ that minimizes the
weighted training error rate $\epsilon_f = \sum_{i=1}^{n} w_i
\text{Prob}(f(\myket{\psi_i})\neq y_i)$.
\end{qlt}

Once again, if we are in the idealized situation where we know the
classical descriptions of the states (learning class
$\mathsf{L}_{qu}^{cl}$), their weights can be directly incorporated
in the description of the density matrices of their classes. In this
scenario, the following reduction formalizes how to solve the
weighted binary classification task given the access to an Helstrom
oracle (version (1)).

\begin{red}[Reduction from weighted binary classification to
standard binary classification (via Helstrom
oracle)]\label{first_weighted_red} Given the access to an Helstrom
oracle that takes as inputs the description of the density matrices
$\rho_-$ and $\rho_+$ (and their \emph{a priori} probabilities $p_-$
and $p_+$), it is possible to \emph{reduce the task of weighted
binary classification to the task
of standard binary classification}.\\
\textbf{Training cost}: null.\\
\textbf{Classification cost}: $\Theta(1)$.
\begin{proof}
The weight $w_i$ of a particular state can be converted to a
probability $p_i$ reflecting its importance by setting
\begin{equation}
p_i=\frac{w_i}{\sum_{j=1}^{n}w_j}.
\end{equation}
Let $\hat{p}_-$, be the new \emph{a priori} probability of the
negative class, which is equal to
\begin{equation}
\hat{p}_- = \sum_{i=1}^{n}p_iI\{y_i=-1\}
\end{equation}
and $\hat{p}_+$, its complementary probability such that $\hat{p}_-
+ \hat{p}_+ = 1$. Theorem~\ref{helstrom_weighted} demonstrates that
the Helstrom measurement which discriminates between the density
matrices in which the weights are incorporated is precisely the POVM
which minimizes the weighted error. Therefore, it suffices to call
the Helstrom oracle with inputs
\begin{equation}
\hat{\rho}_-=\sum_{i=1}^{n}p_iI\{y_i=-1\}\myket{\psi_i}\mybra{\psi_i}
\end{equation}
and
\begin{equation}
\hat{\rho}_+=\sum_{i=1}^{n}p_iI\{y_i=+1\}\myket{\psi_i}\mybra{\psi_i}
\end{equation}
(with \emph{a priori} probabilities $\hat{p}_-$ et $\hat{p}_+$).
This reduction makes only one call to the Helstrom oracle and
requires only one copy of the unknown quantum state at
classification.
\end{proof}
\end{red}
\begin{thm}[Helstrom measurement minimizing the weighted
error]\label{helstrom_weighted} The Helstrom measurement which
minimizes the training error between $\hat{\rho}_-=\sum_{i=1}^{n}p_i
I\{y_i=-1\}\myket{\psi_i}\!\mybra{\psi_i}$ and
$\hat{\rho}_+=\sum_{i=1}^{n}p_i
I\{y_i=+1\}\myket{\psi_i}\!\mybra{\psi_i}$ (with \emph{a priori}
probabilities $\hat{p}_-$ and $\hat{p}_+$) is also the POVM which
minimizes the weighted classification error on the quantum training
dataset $D_n
=\{(\myket{\psi_1},y_1,w_1),\ldots,(\myket{\psi_n},y_n,w_n)\}$.
\end{thm}
\begin{proof}
The Helstrom measurement is the POVM $f$ that minimizes the
discrimination error between $\hat{\rho}_-$ and $\hat{\rho}_+$. This
POVM can be decomposed into two elements $\Pi_-$ et $\Pi_+$ which
both correspond to positive semi-definite matrices such that $\Pi_-
+ \Pi_+ = \mathsf{I}$, where $\mathsf{I}$ is the identity matrix.
Therefore, we have:
\begin{equation} \epsilon_{Hel} = \min_f
(\Tr(\Pi_-\hat{\rho}_+) + \Tr(\Pi_+\hat{\rho}_-))
\end{equation}
that can also be express as
\begin{eqnarray}
\epsilon_{Hel} = \min_f
\sum_{i=1}^{n}p_iI\{y_i=+1\}\Tr(\Pi_-\myket{\psi_i}\mybra{\psi_i})+
\\
\sum_{i=1}^{n}p_iI\{y_i=-1\}\Tr(\Pi_+\myket{\psi_i}\mybra{\psi_i})
\end{eqnarray}
and that simplifies to
\begin{equation}
\epsilon_{Hel} = \min_f
\left(\sum_{i=1}^{n}p_i\text{Prob}(f(\myket{\psi_i})\neq y_i)\right)
\end{equation}
which is the same as minimizing the weighted training error:
\begin{equation}
\epsilon_{opt}=
\min_f\left(\sum_{j=1}^{n}w_j\times\epsilon_{Hel}\right)
\end{equation}
\begin{equation}
\epsilon_{opt}=\min_f\left(\sum_{j=1}^{n}w_j\sum_{i=1}^{n}p_i\text{Prob}(f(\myket{\psi_i})\neq
y_i)\right)
\end{equation}
\begin{equation}
\epsilon_{opt}= \min_f
\left(\sum_{i=1}^{n}w_i\text{Prob}(f(\myket{\psi_i})\neq y_i)\right)
\end{equation}
As this POVM is optimal, it automatically implies that its regret is
zero.
\end{proof}

\begin{algorithm}
\caption{$\mathsf{rejection\_sampling}$($D_n \in
\mathsf{L}_{qu}^{\otimes \Theta(t_{bin})}$)}
\label{rejection_sampling}
\begin{algorithmic} \raggedright
\STATE Choose a constant $c$ greater than any weight $w$
  \FOR{each state $\myket{\psi_i}^{\otimes \Theta(t_{bin})}$}
  \STATE Flip a coin which has a \emph{bias}
  of $\frac{w_i}{c}$
  \IF{the result is ``tails''}
  \STATE Keep the copies of the state
  \ELSE \STATE Put them aside
  \ENDIF
  \ENDFOR
\STATE \textbf{Return} the new generated distribution
$\widetilde{D}$
\end{algorithmic}
\end{algorithm}

\begin{algorithm}
\caption{$\mathsf{costing\_training}$($D_n \in
\mathsf{L}_{qu}^{\otimes \Theta(Tt_{bin})}$)}
\label{costing_training}
\begin{algorithmic} \raggedright
 \FOR{$j=1$ to $T$}
  \STATE Call $\mathsf{rejection\_sampling}$($D_n \in
\mathsf{L}_{qu}^{\otimes \Theta(t_{bin})}$) to obtain
  $\widetilde{D}_j$
  \STATE Call the Helstrom oracle on $\widetilde{D}_j$ to learn the binary classifier $f_j$
  \ENDFOR
\STATE \textbf{Return} the final classifier
$f=\text{majority}(f_1,\ldots,f_T)$
\end{algorithmic}
\end{algorithm}

In the case where only a finite number of copies of each quantum
state is accessible, but we know a way of producing an efficient
binary classifier (such as the Helstrom oracle, version 2), then the
\emph{costing reduction}~\cite{costingreduction} enables to reduce
weighted binary classification to standard binary classification.
This reduction proceeds via a \emph{rejection sampling} mechanism
(Algorithm~\ref{rejection_sampling}) and the \emph{aggregation of
several classifiers} (Algorithm~\ref{costing_training}), and
generates an ensemble of $T$ binary classifiers, where $T$ is a
small constant chosen independently from $D_n$.

The output of the final classifier is simply a majority vote on the
outputs of the individual classifiers. The number of copies of the
unknown state $\myket{\psi_?}$ used by the final classifier is a
constant $\Theta(T)$, corresponding to the number of binary
classifiers forming the aggregated classifier
(Algorithm~\ref{costing_classification}). It is clear that the more
evaluations are done, the more accurate the classification will be,
but more copies of $\myket{\psi_?}$ will be needed.

\begin{red}[Reduction from weighted binary classification to
standard binary classification (via
costing~\cite{costingreduction})] Given the access to an Helstrom
oracle (version 2) and a quantum training dataset
$D_n\in\mathsf{L}_{qu}^{t_{bin}}$, it is possible to \emph{reduce
the task of weighted binary classification to the task
of standard binary classification}.\\
\textbf{Training cost}: $\Theta(Tt_{bin})$.\\
\textbf{Classification cost}: $\Theta(T)$.
\begin{proof}
During the training, the algorithm $\mathsf{costing\_training}$
calls the Helstrom oracle $T$ times, for a constant $T$ chosen
independently from the training dataset $D_n$. The training cost is
therefore $\Theta(Tt_{bin})$, which corresponds to the number of
calls to the Helstrom oracle multiplied by $t_{bin}$ the number of
copies of each state required at each call. As each call to the
Helstrom oracle produces a classifier, the classification cost is
$\Theta(T)$, which requires to use a copy of the unknown state
$\myket{\psi_?}$ for each generated classifier. The analysis of the
costing reduction~\cite{costingreduction} demonstrates that the
average of the standard training errors that minimize the individual
classifiers $f_1,\ldots,f_T$ on the distributions
$\widetilde{D}_1,\ldots,\widetilde{D}_T$ is the same as indirectly
minimizing the weighted training error of the global classifier $f$,
which means:
\begin{equation}
\epsilon_f \sim \min_f\sum_{i=1}^{n} w_i
\text{Prob}(f(\myket{\psi_i})\neq y_i)
\end{equation}
\end{proof}
\end{red}

\begin{algorithm}
\caption{
$\mathsf{costing\_classification}$($\myket{\psi_?}^{\otimes
\Theta(T)},f=(f_1,\ldots,f_T)$)} \label{costing_classification}
\begin{algorithmic} \raggedright
 \FOR{$j=1$ to $T$}
  \STATE Measure $y_j = f_j(\myket{\psi_?})$
  \ENDFOR
\STATE \textbf{Return} $y_?=\text{majority}(y_1,\ldots,y_T)$
\end{algorithmic}
\end{algorithm}

The quantum version of rejection sampling
(Algorithm~\ref{rejection_sampling}) has the additional benefit of
``saving'' some copies of the quantum states during the generation
of the distribution biased according to their weights. Indeed, the
states having a low weight have a higher probability of not being kept
in the new generated distribution. Therefore, these states can be
put aside and used later, for instance during another step of
rejection sampling.

\section{Multiclass classification}
\label{multiclass classification}

In the multiclass version of classification, each state is labeled
after a class chosen among $k$ possible ones, for $k>2$. The goal is
to build a classifier $f$ which, given a finite number of copies of
an unknown state $\myket{\psi_i}$, can predict its class $y_?$ with
a good accuracy.

\begin{qlt}[(Quantum) multiclass classification]
\label{binary_class}
:\\
\textbf{Input}: \mbox{$D_n
=\{(\myket{\psi_1},y_1),\ldots,(\myket{\psi_n},y_n)\}$}, a quantum
training dataset, where
$\myket{\psi_i} \in \mathbb{C}^{2^d}$ and $y_i \in \{1,\ldots,k\}$.\\
\textbf{Output}: A POVM acting as a multiclass classifier $f$ that
can predict the class $y_?$ of an unknown quantum state
$\myket{\psi_?}$.\\ \textbf{Goal}: Construct a multiclass classifier
$f$ that minimizes the training error rate $\epsilon_f = \frac{1}{n}
\sum_{i=1}^{n} \text{prob}(f(\myket{\psi_i})\neq y_i)$.
\end{qlt}
Moving from the binary to the multiclass case is far from being
trivial, and very few things are known for the case where the number
of classes $k>2$. In particular even for three classes, the exact
form of the optimal POVM that can distinguish between these three
classes given a single copy of a state is not known. However, we
will see in Section~\ref{section_pgm} that if we know the classical
description of the states, it is possible to design a measure
(called the \emph{Pretty Good
Measurement}~\cite{prettygoodmeasurement}), whose error is bounded
by the square root of the error of the optimal POVM.

The following sections describe different training and
classification strategies for the cases where we have access to a
number of copies of the unknown quantum state $\myket{\psi_?}$ to
classify which is:\\ - linear in $n$, the number of states in
$D_n$ (Section~\ref{class_state_identification}).\\
- linear in $k$, the number of classes in $D_n$
(Section~\ref{section_oneagainstall}).\\
- logarithmic in $k$ (Section~\ref{section_binary_tree}).\\
- a single copy or possibly a constant number of them
(Section~\ref{section_pgm}).

\subsection{Classification via state
identification}\label{class_state_identification}

The most direct way of recognizing the class of a state is to
\emph{identify exactly this state}. Once the state is identified,
this information allows also to recover directly its class (unless
there are two, or more, states that are identical but labeled with
different classes). If $\Theta(n)$ copies of the unknown quantum
state $\myket{\psi_?}$ are available, the \textsf{Control-Swap}
test~\cite{symmetrisation:96,fingerprinting:01} can be used between
this state and each of the state of the training dataset
$D_n\in\mathsf{L}_{qu}^{\Theta(1)}$. This method does not require
any effort during the training, all the work being done at
classification time (therefore it corresponds to a learning strategy
type (2), Section~\ref{binary_classification}). This learning
strategy can be seen as the quantum analogue of the one-nearest
neighbours. Indeed, for each state of the training dataset, we
search the one which is the closest/ the most similar (in the sense
of fidelity) from the unknown quantum state. Unless there are two
quantum states in $D_n$ that are identical but labeled with two
different classes, this method is guarantee to have a null
classification error (and therefore a null regret). The following
algorithm formalize this method.

\begin{algorithm}
\caption{
$\mathsf{classification\_via\_identification}$($\myket{\psi_?}^{\otimes
\Theta(n)}$, $D_n \in \mathsf{L}_{qu}^{\otimes \Theta(1)}$)}
\label{identification_classification}
\begin{algorithmic} \raggedright
 \FOR{$i=1$ to $n$}
  \STATE Measure the fidelity between $\myket{\psi_?}$ and $\myket{\psi_i}$ by using the \textsf{Control-Swap} test which gives an estimate of $\Fid(\myket{\psi_?},\myket{\psi_i})$
  \ENDFOR
\STATE \textbf{Return} the class $y_j$ of the state $\myket{\psi_j}$
whose fidelity with the unknown quantum state is maximal $\argmax_j
\Fid(\myket{\psi_?},\myket{\psi_j})$
\end{algorithmic}
\end{algorithm}

\begin{thm}[Classification via state identification]
The algorithm $\mathsf{classification\_via\_identification}$
classifies an unknown quantum state $\myket{\psi_?}$ with a null
classification error given $\Theta(n)$ copies of this state and
$\Theta(1)$ copies of each state of $D_n$.
\begin{proof}
Each \textsf{Control-Swap} test require a constant number of copies
and as we estimate the similarity between $\myket{\psi_?}$ and all
the $n$ states of $D_n$, the global cost of
$\mathsf{classification\_via\_identification}$ will be $\Theta(n)$
copies of the unknown state and $\Theta(1)$ of each state of the
training dataset. Moreover, if there are not two states in $D_n$
that are identical but labeled with two different classes, the
algorithm is guarantee to obtain of null classification error (which
implies a null regret).
\end{proof}
\end{thm}

If we want to base the prediction of the class of $\myket{\psi_?}$
on its $k$ nearest neighbours instead of only its nearest neighbour,
the Algorithm~\ref{identification_classification} can be easily
adapt to base its prediction on a majority vote of their classes
(the training and classification cost remain unchanged). An
interesting avenue of research is to design a quantum equivalent to
classical data structures that can be used to facilitate the search
for nearest neighbours, such as the \emph{k}d-trees~\cite{kdtrees}
for instance. Quantumly, the main purpose of such a structure would
be to retrieve the nearest neighbours of an unknown state by
consuming less copies than require with the direct na\"{\i}ve method
(for instance by using a number of copies logarithmic in $n$ and
linear in $c$ the number of neighbours considered). If we do not
know the classical description of the states, the construction of
this data structure may have a non-negligible training cost.

\subsection{One-against-all reductions}\label{section_oneagainstall}

\begin{algorithm}
\caption{
$\mathsf{one\_against\_all\_training}$($D_n \in
\mathsf{L}_{qu}^{\otimes \Theta(kt_{bin})}$)}
\label{one_against_all_training}
\begin{algorithmic} \raggedright
 \FOR{$j=1$ to $k$}
 \STATE Initialize $D^{(j)}$ as the empty dataset
 \FOR{$i=1$ to $n$}
  \STATE Add the example $(\myket{\psi_i}^{\otimes
\Theta(t_{bin})},1-2I\{y_i=j\})$ to $D^{(j)}$
  \ENDFOR
  \STATE Call the Helstrom oracle on the dataset $D^{(j)}$ to learn a binary classifier $f_j$  that discriminates between the class $j$
  and the union of all the other classes
  \ENDFOR
\STATE \textbf{Return} the ensemble of binary classifiers
$f_1,\ldots,f_j$
\end{algorithmic}
\end{algorithm}

The main idea of the \emph{one-against-all
reduction}~\cite{oneagainstall} is to train a binary classifier for
each of the $k$ classes. Each of this binary classifier
\emph{discriminates between its own class and the union of all the
other classes}. This reduction can be adapted in a straightforward
manner to the quantum context by constructing for each class a POVM
acting as a binary classifier, which discriminates between the
density matrix of this class and the statistical mixture composed of
the density matrices of the other classes. We will say that a
classifier ``click'' if it predicts that the unknown state
$\myket{\psi_?}$ belongs to its own class, and that it ``does not
click'' otherwise. Given the access to an Helstrom oracle, it is
possible to reduce the multiclass classification to the standard
binary case by using the following training and classification
algorithms (Algorithms~\ref{one_against_all_training}
and~\ref{one_against_all_classification}).

\begin{algorithm}
\caption{
$\mathsf{one\_against\_all\_classification}$($\myket{\psi_?}^{\otimes
k}$)} \label{one_against_all_classification}
\begin{algorithmic} \raggedright
 \FOR{$j=1$ to $k$}
 \STATE Apply a binary classifier $f_j$ on $\myket{\psi_?}$ to obtain the prediction whether or not this state belongs to the class $j$
  \ENDFOR
  \IF{only one classifier ``has clicked''}
\STATE \textbf{Return} the class associated with the classifier which has ``clicked''
 \ELSE \IF{several classifiers have ``clicked''}
 \STATE \textbf{Return} a class chosen at random among all the classifiers which have ``clicked''
 \ELSE
 \STATE \textbf{Return} a class chosen uniformly at random among the $k$ classes
 \ENDIF
 \ENDIF
\end{algorithmic}
\end{algorithm}

\begin{red}[Reduction from multiclass classification to standard binary classification (via one-against-all)]
Given the access to an Helstrom oracle and a quantum training dataset $D_n\in\mathsf{L}_{qu}^{\Theta(k t_{bin})}$, it is possible to \emph{reduce the multiclass classification task to the standard binary classification via a one-against-all reduction}.\\
\textbf{Training cost}: $\Theta(k t_{bin})$.\\
\textbf{Classification cost}: $\Theta(k)$.
\begin{proof}
The algorithm $\mathsf{one\_against\_all\_training}$ calls the
Helstrom oracle a number of times which is linear in the number of
classes $k$, and each call consumes a number of copies of each state
of $D_n$ in $\Theta(t_{bin})$. Therefore, the training cost of this
reduction is $\Theta(k t_{bin})$. Regarding the classification, we
need to sacrifice a copy of the unknown state $\myket{\psi_?}$ for
each of the $k$ binary classifiers generated, which leads to a total
cost of $\Theta(k)$.

Regarding the analysis of the error of this reduction, let
$\epsilon_j$ be the error of the classifier of class $j$. The worst
situation that can happen is that the classifier of the ``good
class'' does not click (which corresponds to a false negative). In
this situation and if no other classifier has clicked, we choose the
class to predict uniformly at random, which lead to an error with
probability $\frac{k-1}{k}$. In the case of false positives, where
$c$ classifiers click when they should not, the error rate will be
only $\frac{c}{c+1}$ because we will choose at random among the
classifiers which have reacted. As each binary classifier $f_j$
leads to an error rate of $\frac{k-1}{k}$ in the worst case with
probability $p_j\epsilon_j$ (where $p_j$ is the \emph{a priori}
probability of class $j$) and there are $k$ binary classifiers, the
global error of the classifier will be upper bounded by
$\frac{k-1}{k}\sum_{j=1}^{k}p_i\epsilon_j$, which simplifies itself
to $(k-1)\epsilon$ if all the classes have the same \emph{a priori}
probability $\frac{1}{k}$ and the same error rate $\epsilon$ for all
the binary classifiers. (This reduction does not seem to offer any
guarantee for the regret.)
\end{proof}
\end{red}
\begin{rem}[Difficulty of intermediary learning situations generated by the reduction]
Nothing guarantee \emph{a priori} than the intermediary learning
situations generated by the reduction (for instance here the $k$
binary classification) are easy to solve. Indeed, even if the access
to the Helstrom oracle guarantee than the $k$ binary classifiers
will be optimal for their respective classification settings, it is
possible than the observed average error will be important. In the
quantum case, it can happen for instance than the trace distance
between the density matrix of a class and the mixture composed of
the union of all the other classes is low (which implies that they
are difficult to distinguish). If we have a complete classical
knowledge of the quantum states instead a simply deriving an upper
bound of the error of the global classifier, a finer analysis will
reveal the exact training error of this classifier.
\end{rem}
A weighted variant of the reduction, called \emph{weighted
one-against-all}~\cite{weightedoneagainstall}, offers a better upper
bound in terms of error than the basic version. This variant
exploits the fact than false negatives (not detecting the true
class) are more damageable to the error of global classifier than
false positives (predicting the wrong class). In practice, this
means that a datapoint will have a higher weight during the
construction of the classifier of its class. The algorithm proceeds
by reducing the multiclass classification to weighted binary
classification and then use the costing
reduction~\cite{costingreduction} to reduce the weighted binary
classification to the standard binary classification. The main
advantage of the weighted version of this reduction is that it
offers a guarantee on the error bound of the global classifier of
$\frac{k}{2}\epsilon$, for $\epsilon$ the average error of the
binary classifiers generated, which is divided by two compared to
the basic version. In this case, the training cost of the weighted
version of the one-against-all reduction is $\Theta(kTt_{bin})$
where $k$ is the number of classes, $T$ the constant number of
classifiers generated by the costing reduction and $t_{bin}$ the
number of copies used by each call of the Helstrom oracle. The
classification cost will be $\Theta(kT)$. Quantumly, if we know the
classical description of the states of the training dataset
($D_n\in\mathsf{L}_{qu}^{cl}$), we can replace the costing reduction
via the reduction using the Helstrom oracle
(Reduction~\ref{first_weighted_red}) which results in a training
cost of $\Theta(kt_{bin})$ and a classification cost of $\Theta(k)$.

\subsection{Binary tree reductions}\label{section_binary_tree}

Another way of solving the multiclass version is to build a binary
tree where \emph{each node is a binary classifier which
discriminates between two subsets of classes} and where \emph{the
leaves are labeled after a specific class}. The root contains the
set of all classes and use a binary classifier to divide this set
into two subsets of classes of approximatively same size. To
classify an unknown state, we start from the root and we go down the
tree according to the output of the binary classifier observed at
each node until we reach a leaf, in which case we predict the class
associated to this leaf. There are several ways of building the
binary tree (for instance in a bottom-up or top-down fashion), which
might lead to a different global error of the final classifier. The
Algorithms~\ref{binary_tree_training}
and~\ref{binary_tree_classification} detail a possible way of
constructing recursively the binary tree from the root to the
leaves, and then use it for classification.

\begin{algorithm}
\caption{
$\mathsf{binary\_tree\_training}$($D_n\in\mathsf{L}_{qu}^{\otimes
\Theta(t_{bin}\log k)}$)} \label{binary_tree_training}
\begin{algorithmic} \raggedright
\IF{all the states in $D_n$ belongs to the same class} \STATE
\textbf{Create} a leaf labeled according to this class \STATE
\textbf{Return} \ENDIF \STATE \textbf{Choose} at random two subsets
of classes $Y_a$ and $Y_b$ among $D_n$ such that $|Y_a|\approx
|Y_b|$ \STATE \textbf{Separate} the training dataset $D_n$ into two
subsets $D_a$ and $D_b$ according to the two subsets of classes
$Y_a$ and $Y_b$ (let $\rho_a$ be the density matrix representing the
subset $D_a$ and $\rho_b$ the density matrix representing the subset
$D_b$) \STATE \textbf{Call} the Helstrom oracle to learn the binary
classifier $f_{(\rho_a,\rho_b)}$ which distinguishes between the two
density matrices $\rho_a$ and $\rho_b$ \STATE \textbf{Create} a node
in the binary tree whose test corresponds to the binary classifier
$f_{(\rho_a,\rho_b)}$ \STATE \textbf{Call}
$\mathsf{binary\_tree\_training}$($D_a$) \STATE \textbf{Call}
$\mathsf{binary\_tree\_training}$($D_b$)
\end{algorithmic}
\end{algorithm}

\begin{algorithm}
\caption{$\mathsf{binary\_tree\_classification}$($\myket{\psi_?}^{\otimes
\Theta(\log k)}$, a classifier $f$ which is a binary classification
tree)} \label{binary_tree_classification}
\begin{algorithmic} \raggedright
\STATE \textbf{Start} the traversal of the tree at the root \WHILE{a
leaf is not reach} \STATE Use a copy of the state $\myket{\psi_?}$
in the binary classifier corresponding to the current node \IF{the
classifier predicts the negative class} \STATE Go down the tree on
the left \ELSE \STATE Go down the tree on the right\ENDIF \ENDWHILE
\STATE \textbf{Return} the class labeled at this leaf
\end{algorithmic}
\end{algorithm}

\begin{red}[Reduction of multiclass classification to standard
binary classification (via binary
tree)]\label{multiclass_binary_tree_red}
Given the access to an Helstrom oracle and a quantum training dataset $D_n\in\mathsf{L}_{qu}^{\Theta(t_{bin}\log k)}$, it is possible to \emph{reduce the multiclass classification task to the standard binary classification via a binary tree reduction}.\\
\textbf{Training cost}: $\Theta(t_{bin}\log k)$.\\
\textbf{Classification cost}: $\Theta(\log k)$.
\begin{proof}
During the construction of the binary tree, the Helstrom oracle is
called a number of times which is directly proportional to the
number of nodes in the tree. However, each call to the oracle splits
the dataset into two subsets (whose sum of sizes is equal to that of
the original training dataset), which implies that at each level of
the tree the number of copies of each quantum state used by the
different calls of the Helstrom oracle is $\Theta(t_{bin})$. The
global cost of the training is therefore $\Theta(t_{bin}\log k)$
because the depth of the tree is $\Theta(\log k)$ (for $k$ the
number of classes) as it is built to be balanced. The classification
cost is also directly proportional to the depth of the tree and is
$\Theta(\log k)$.

The global error of the final binary tree classifier is the sum of
probability for each class of having a error in the path going from
the root to the leaf of this class which is upper bounded by
$\epsilon \log k$, if for simplification we suppose that all the
classes are equiprobable and that all binary classifiers have the
same error $\epsilon$. Indeed in this case, an error can occur with
probability $\epsilon$ at each node traversed which implies that the
global error maybe $\epsilon \log k$ in the worst case.
\end{proof}
\end{red}

\begin{cor}[State identification]
Let $D_n$ be a quantum training dataset composed of $n$ pure states
such that there are not two identical states in $D_n$. A POVM exists
that can identify the index of an unknown quantum state
$\myket{\psi_?}$ chosen at random among the states of $D_n$ with a
non-trivial accuracy given $\Theta(\log n)$ copies of this state.
\begin{proof}
The proof is relatively direct, it simply involves setting $k=n$,
which means assigning a different class to each of the $n$ points of
the quantum dataset $D_n$, and applying the
Reduction~\ref{multiclass_binary_tree_red}.
\end{proof}
\end{cor}

If we are in the situation where we have a complete knowledge of the
states of the training set ($D_n\in\mathsf{L}_{qu}^{cl}$), it is
possible to choose the two subsets of classes such that they
maximize the trace distance between the two density matrices of
these subsets. In this case, it is possible to build the tree from
the root to the leaves by splitting the dataset into two subsets
which maximize the trace distance. Another way of growing the tree
is by starting from the leaves to the root, where at each level we
pair the classes that are the easiest to distinguish. In particular,
a reduction called ``\emph{filter tree}''\cite{filtertree} exists
which reduce the multiclass classification to the standard binary
classification (via weighted binary classification and the
\emph{costing} reduction~\cite{costingreduction}). This reduction
builds a multiclass classifier which has the form of a binary tree
by starting from the leaves and guarantee that the error of this
classifier is upper bounded by $\epsilon \log k$, for $k$ the number
of classes and $\epsilon$ the average error of the binary
classifiers generated. The strength of this reduction is that it
offers a similar guarantee for the regret (which was not the case of
the algorithm $\mathsf{binary\_tree\_training}$ presented
previously). The regret of the multiclass classifier will be at most
$r \log k$, for $r$ the average regret of the binary classifiers.

\subsection{Pretty good measurement}\label{section_pgm}

If we know the classical description of the states
\mbox{($D_n\in\mathsf{L}_{qu}^{cl}$)}, a general measurement
strategy exists, called the ``\emph{Pretty Good
Measurement}''\footnote{This measure is sometimes called
``\emph{square-root measurement}'' in the literature due to the
explicit form of this POVM.}~\cite{prettygoodmeasurement}, which
enables us to build a classifier, which given a single copy of an
unknown state $\myket{\psi_?}$, can predict the class of this state
with an error bounded by the square root of the error of the optimal
classifier.
\begin{thm}[Error rate of the Pretty Good Measurement~\cite{errorpgm}]
Given the classical description of $k$ density matrices
$\rho_1,\ldots,\rho_k$, it is possible to build a POVM, called the
\emph{Pretty Good Measurement}, whose error rate $\epsilon_{PGM}$ to
distinguish between these $k$ mixed states, given a single copy
$\rho_?$ of one of these states, is in the worst case quadratically
higher that the error $\epsilon_{opt}$ that would have the optimal
POVM. Formally, this means that:
\begin{equation}
\epsilon_{opt} \leq \epsilon_{PGM} \leq \sqrt{\epsilon_{opt}}
\end{equation}
\end{thm}
\begin{cor}[Bound on the regret of the Pretty Good Measurement]
The regret of the Pretty Good Measurement is bounded by:
\begin{equation}
r_{PGM} \leq \sqrt{\epsilon_{opt}} - \epsilon_{opt}
\end{equation}
\end{cor}
Montanaro~\cite{discriminationrandomstates} proved that the error of
the Pretty Good Measurement is always smaller than that of the
prediction strategy that does not even measure the state, but rather
chooses one of the classes at random according to their \emph{a
priori} probabilities. He also derived an upper bound on the error
of the Pretty Good Measurement which depends on the fidelity between
each pair of states forming
the training dataset. This bound is:
\begin{equation}\label{upper_bound_pgm}
\epsilon_{PGM}\leq 1 -
\frac{1}{n}\sum_{i=1}^{n}\frac{1}{\sum_{j=1}^{n}\Fid(\myket{\psi_i},\myket{\psi_i})}
\end{equation}
\begin{defin}[Similarity matrix of a quantum training dataset]
A \emph{similarity matrix}\footnote{The similarity matrix is often
called \emph{Gram matrix} in the literature, especially in classical
ML.} $S_n$ of a quantum training dataset containing $n$ states is a
matrix of size $n$ by $n$, where each entry $S(i,j)$ of the matrix
(for $i,j\in\{1,\ldots,n\}$) contains an estimate of the fidelity
between the state $\myket{\psi_i}$ and the state $\myket{\psi_j}$.
\end{defin}

It follows directly from the symmetry property of the fidelity, that
the similarity matrix is a \emph{symmetric matrix}. An efficient
algorithm exists to compute this matrix, which requires only a
number of copies of each state that is linear in $n$, the number of
states in the quantum training dataset. The
Algorithm~\ref{compute_similarity_matrix} formalizes how to compute
the similarity matrix for a quantum dataset $D_n$.

\begin{algorithm}
\caption{
$\mathsf{similarity\_matrix\_computation}$($D_n \in
\mathsf{L}_{qu}^{\otimes \Theta(en)}$)}
\label{compute_similarity_matrix}
\begin{algorithmic} \raggedright
\FOR{$i=1$ to $n$} \STATE $S(i,i)=1$ \ENDFOR
 \FOR{$i<j$}
  \STATE Estimate the fidelity between the two states $\myket{\psi_i}$
  and
  $\myket{\psi_j}$ by using the \textsf{C-Swap} test $e$ times
  \STATE Set the estimate of $\Fid(\myket{\psi_i},\myket{\psi_j})$ to be equal to \mbox{$1- \frac{2 \times
\#\myket{1}}{e}$} (where $\#\myket{1}$ represents the number of
times where the result $\myket{1}$ has been observed)
  \STATE Update $S(i,j)=S(j,i)=\Fid(\myket{\psi_i},\myket{\psi_j})$
\ENDFOR \STATE \textbf{Return} $S_n$ the computed similarity matrix
\end{algorithmic}
\end{algorithm}
\begin{thm}[Computation of the similarity matrix]
It is possible to compute the similarity matrix of a quantum dataset
$D_n$ with a precision $\epsilon$, for $\epsilon=\frac{1}{e}$, from
$\Theta(en)$ copies of each state.
\begin{proof}
For each pair of states ($\myket{\psi_i}$,$\myket{\psi_j}$) of the
training dataset $D_n$, the \textsf{Control-Swap} test allows us to
estimate the fidelity between these states with a precision
$\epsilon$, where $\epsilon=\frac{1}{e}$ for $e$ the number of
copies used during the test. As the matrix $S_n$ is symmetric, the
number of entries to estimate is
$\Theta(\frac{n(n-1)}{2})=\Theta(n^2)$. Therefore for each state
$\myket{\psi}$, we will need $\Theta(e)$ copies for each of the $n$
\textsf{Control-Swap} tests where this state appears, which makes a
global cost of $\Theta(en)$ copies per state.
\end{proof}
\end{thm}
\begin{cor}[Upper bound on the error of the Pretty Good Measurement]
Given $\Theta(n)$ copies of each state of a quantum training dataset
$D_n$, it is possible to compute an upper bound on the error of that
Pretty Good Measurement will make on $D_n$.
\begin{proof}
The proof is straightforward, we only need to apply the algorithm
$\mathsf{similarity\_matrix\_computation}$ and evaluate the
formula~\ref{upper_bound_pgm} by using the estimate the fidelity
between each pair of states from the corresponding entries of the
similarity matrix.
\end{proof}
\end{cor}
Montanaro also gave another upper bound on the error of the Pretty
Good Measurement which depends directly on the eigenvalues of the
similarity matrix $S_n$. Let $\lambda_i$, be the $i^{\text{th}}$
eigenvalue of the similarity matrix. The error of the Pretty Good
Measurement is bounded from above by:
\begin{equation}
\epsilon_{PGM} \leq 1 -
\frac{1}{n}\left(\sum_{i=1}^{n}\sqrt{\lambda_i}\right)^2
\end{equation}
This bound can also be explicitly computed from the similarity
matrix $S_n$ by diagonalizing it to extract the eigenvalues.

Regarding a lower bound of the Pretty Good Measurement, a recent
bound~\cite{lowerboundpgm}, also due to Montanaro, proved that this
error is bounded from below by:
\begin{equation}\label{lower_bound_pgm}
\epsilon_{PGM} \geq \sum_{i=1}^{n}\sum_{j=i}^{n}p_i
p_j\Fid(\myket{\psi_i},\myket{\psi_j}),
\end{equation}
where $p_i$ and $p_j$ are the \emph{a priori} probabilities of the
states $\myket{\psi_i}$ and $\myket{\psi_j}$. In the situation where
all these states are equiprobable, we can replace all the
probabilities by the $\frac{1}{n}$ in the
formula~\ref{lower_bound_pgm}. Here also the bound can be directly
estimated from the similarity matrix $S_n$.

Intuitively, this bounds seem to indicate that the fidelity between
pair of states is a sufficient measure to assess the difficulty of
distinguishing the states of the training dataset. This intuition is
wrong, indeed Jozsa and Schlienz~\cite{jozsaschlienzparadox} have
shown that there exist situations where the fidelity between pair of
states in the quantum dataset $D_n$ is low (which means it is easy
to discriminate one of state from the other), while at the same
time it is impossible to distinguish efficiently in a global manner
one state from all the other states.

To summarize, \emph{it is possible to bound the error that the
Pretty Good Measurement would realize even without explicitly
constructing it}. Indeed, we can bound the error of the Pretty Good
Measurement given a linear number of copies of each state of the
quantum dataset, whereas if we want to build explicitly the POVM
corresponding to this measure all the techniques currently known
seems to require to know a classical description of the states
(which requires an exponential number of copies in the number of
qubits if we use the tomography in the case of an unknown state). An
important avenue of research is whether or not it is possible to
design a learning algorithm that ``learns'' an approximate version
of the Pretty Good Measurement (in the same sense as the Helstrom
oracle) from a finite number of copies of each state of $D_n$.
\begin{conj}[Amount of information necessary to learn the Pretty
Good Measurement] The minimal number of copies $t_{PGM}$ of each
state of the quantum training dataset $D_n$ necessary to ``learn'' a
circuit that could implement a non-trivial approximation of the
Pretty Good Measurement is polynomial in $n$ the number of quantum
states in $D_n$ and $k$ the number of classes.
\end{conj}

\section{Discussion and conclusion}\label{discussion_and_conclusion}

The following table summarizes the training/learning and the
classification cost of the different learning tasks and reductions
that we have seen in this paper. The \emph{binary classification is
the main learning primitive} as the weighted and the multiclass
classification can be reduced to it via the Helstrom oracle.

\begin{table}[h]
\begin{tabular}{|l|l|l|}
\hline
Learning task & Training cost & Classification cost \\
\hline
Binary classification & $\Theta(t_{bin})$ & $\Theta(1)$ \\
\hline
Weighted binary classification & & \\
(reduction via Helstrom oracle) & $\Theta(t_{bin})$ & $\Theta(1)$\\
(costing reduction) & $\Theta(Tt_{bin})$ & $\Theta(T)$\\
\hline
Multiclass classification &  & \\
(state identification via \textsf{Control-SWAP} test) & $\Theta(1)$& $\Theta(n)$\\
(one-against-all reduction) & $\Theta(kt_{bin})$& $\Theta(k)$\\
(binary tree reduction) & $\Theta(t_{bin}\log k)$ & $\Theta(\log k)$\\
(Pretty Good Measurement) & unknown & $\Theta(1)$\\
(Bound on the error of the PGM) & $\Theta(n)$ & not applicable\\
\hline
\end{tabular}\caption{Table summarizing the training and classification costs of the different quantum learning tasks and reductions seen in this paper.}\label{tab_recap_couts}
\end{table}

In practice, the Helstrom oracle will be implemented by a learning
algorithm, which from a finite number of copies of each state from
the training dataset, outputs a POVM $f$ which can act as a binary
classifier. Contrary to the Helstrom oracle, this algorithm does not
need to be optimal in terms of classification error as long as it
offers a non-trivial precision which is better than simply guessing
randomly the class of the unknown quantum state. Even in this case,
almost all the reductions presented in this paper will work although
the global error of the generated classifier will likely be higher
due to the non-optimality of the constructed POVM. Designing a
learning algorithm as the Helstrom oracle will enable us to estimate
the minimum number of copies $t_{bin}$ of each state of the training
dataset that is necessary to perform the binary classification.

The essence of ML is to learn from data coming from past experience
with the hope of generalizing on new situations in the future. In
this paper, we concentrate on the accurate classification of states
coming from the training dataset $D_n$ but we did not discussed how
this approach could generalize on quantum states unobserved
previously. A natural way of defining that a POVM $f$ acting as a
classifier generalize is if this POVM can recognize the class of a
state that is close to one of the state of the training dataset but
without being identical. The closeness between two pure states can
be defined using the fidelity or other distance measures such as the
\emph{Euclidean distance}.
\begin{defin}[Euclidean distance between pure states~\cite{lemmaBernsteinVazirani}]
The \emph{Euclidean distance} between two pure states
$\myket{\psi}=\sum_{i=1}^{d}\alpha_i\myket{i}$ and
$\myket{\phi}=\sum_{i=1}^{d}\beta_i\myket{i}$ is defined as
$Dist_{L2}(\myket{\psi},\myket{\phi})=\sqrt{\sum_{i=1}^d|\alpha_i -
\beta_i|^2}$.
\end{defin}
Bernstein and Vazirani~\cite{lemmaBernsteinVazirani} have proven
that if two pure states $\myket{\psi}$ and $\myket{\phi}$ of same
dimension are within $\epsilon$ Euclidean distance of each other,
the same measure performed on the two states generates samples from
two distributions which have a total variational distance of at most
$4\epsilon$. Therefore, if two states are close in terms of their
Euclidean distance this give a good indication that a POVM $f$
acting as a classifier will with high probability predicts the same
class for these two states. Future work in this model of doing
machine learning on quantum information include the formalization of
the notion of testing and generalization error, as well as the study
of different models of classical and quantum noise (see for instance
the section 8.3 of~\cite{nielsen:book00} for different forms of
quantum noise) and how they affect the robustness of the quantum
learning algorithms.

ML is a field where it is important to valide experimentally the
performance of a learning algorithm and to compare it to other
existing algorithms. Classically, numerous repositories of datasets
are publicly available such as the repository of the University of
California at Irvine\footnote{\url{http://archive.ics.uci.edu/ml/}}
(\emph{UCI repository}) or the \emph{MNIST database} for the
recognition of
characters\footnote{\url{http://yann.lecun.com/exdb/mnist/}}.
Quantumly, once several learning algorithms have been proposed, it
is also important to test them experimentally on quantum datasets
representing realistic situations that experimentalists are likely
to encounter in their laboratories. The main idea would not be to
create physically these datasets but rather to give access to their
classical descriptions to the community so that anyone who want to
use and experiment with them using their favorite classical
simulator can do it freely. An example of two possible classes could
be for instance entangled state versus separable states. Moreover,
several situations that people encounter in quantum information
processing can be recast naturally as a classification problem, such
as for instance the scenario in quantum cryptography where the
eavesdropper try to maximize his probability of guessing correctly
the class of the state that he has intercepted.

\section*{Acknowledgments}

I would like to thanks Gilles Brassard for enlightening discussions
on the subject and Fr\'{e}d\'{e}ric Dupuis for proofreading an early
version of this paper and his suggestions and comments.


\begin{thebibliography}{99}

\bibitem{learningquantumstates}
Aaronson, S.,
\newblock ``The learnability of quantum states'',
\newblock \emph{Proceedings of the Royal Society A} \textbf{463}(2088), 2007.

\bibitem{learningquantumworld}
A\"{\i}meur, E., Brassard, G. and Gambs, S.,
\newblock ``Machine learning in a quantum world'',
\newblock \emph{Proceedings of the 19th Canadian Conference on Artificial Intelligence (Canadian AI'06)}, pp.~433--444, 2006.

\bibitem{symmetrisation:96}
Barenco, A., Berthiaume, A., Deutsch, D., Ekert, A., Jozsa, R. and
Macchiavello, C.,
\newblock ``Stabilisation of quantum computations by symmetrisation'',
\newblock \emph{SIAM Journal of Computing} \textbf{26}(5), pp.~1541--1557,
1997.

\bibitem{errorpgm}
Barnum, M. and Knill, E.,
\newblock ``Reversing quantum dynamics with near-optimal quantum and classical fidelity'',
\newblock \emph{Journal of Mathematical Physics} \textbf{43}(5), pp.~2097--2106,
2002.

\bibitem{kdtrees}
Bentley, J.L.,
\newblock ``Multidimensional binary search tree used for associative searching'',
\newblock \emph{Communications of the ACM} \textbf{9}(18), pp.
509--517, 1975.

\bibitem{surveystatediscrimination}
Bergou, J., Herzog, U. and Hillery, M.,
\newblock ``Discrimination of quantum states'',
\newblock \emph{Chapter 11: Invited Review Article in Lectures Notes in Physics, vol 649: Quantum State Estimation}, pp.~417--465,
\newblock Springer-Berlin, 2004.

\bibitem{lemmaBernsteinVazirani}
Bernstein, E. and Vazirani, U.,
\newblock ``Quantum complexity theory'',
\newblock \emph{Proceedings of the 25th Annual ACM Symposium on Theory of Computing (STOC'93)}, pp.~11--20,
\newblock 1993.

\bibitem{learningreductions}
Beygelzimer, A., Dani, V., Hayes, T., Langford, J. and Zadrozny, B.,
\newblock ``Error limiting reductions between classification tasks'',
\newblock \emph{Proceedings of the 22th Annual International Conference of Machine Learning (ICML'05)}, pp.~49--56, 2005.

\bibitem{weightedoneagainstall}
Beygelzimer, A., Langford, J. and Zadrozny, B.,
\newblock ``Weighted one-against-all'',
\newblock \emph{Proceedings of the 20th National Conference on Artificial Intelligence (AAAI'05)}, pp.~720--725, 2005.

\bibitem{filtertree}
Beygelzimer, A., Langford, J. and Ravikumar, P.,
\newblock ``Multiclass classification with filter trees'',
\newblock \emph{Unpublished}, 2007.

\bibitem{fingerprinting:01}
Buhrman, H., Cleve, R., Watrous, J. and de Wolf, R.,
\newblock ``Quantum fingerprinting'',
\newblock \emph{Physical Reviews Letters} \textbf{87}(16), article 167902,
2001.

\bibitem{covariantmeasurements:04}
Chiribella, G., Mauro D'Ariano, G., Perinotti, P. and Sacchi, M.,
\newblock ``Covariant quantum measurements which maximize the likelihood'',
\emph{Physical Reviews A} \textbf{70}, article 061205, 2004.

\bibitem{corollaryHolevo}
Cleve, R., van Dam, W., Nielsen, M. and Tapp, A.,
\newblock ``Quantum entanglement and the communication complexity of the inner product'',
\newblock \emph{Proceedings of the First NASA International Conference on Quantum Computing and Quantum Communications}, pp.~61--74, 1999.

\bibitem{confidencebasedmeasurements}
Croke, S., Anderson, E., Barnett, S.M., Gilson, C.R. and Jeffers,
J.,
\newblock ``Maximum confidence quantum measurements'',
\emph{Physical Reviews} \textbf{76}, 2006.

\bibitem{book_pattern_classification}
Duda, R., Hart, P. and Stork, D.,
\newblock \emph{Pattern Classification},
\newblock Wiley-Interscience,
2001.

\bibitem{prettygoodmeasurement}
Hausladen, P., and Wootters, W.K.,
\newblock ``A ``pretty good'' measurement for distinguishing quantum states'',
\emph{Journal of Modern Optics} \textbf{41}, 1994.

\bibitem{helstrom:detection76}
Helstrom, C.\,W.,
\newblock \emph{Quantum Detection and Estimation Theory},
\newblock Academic Press, 1976.

\bibitem{unambiguousbounds:05}
Herzog, U. and Bergou, J.\,A.,
\newblock ``Optimal unambiguous discrimination of two mixed quantum
states'', \emph{Physical Reviews A} \textbf{71}, article 050301,
2005.

\bibitem{holevo:capacity73}
Holevo, A.\,S.,
\newblock ``Bounds for the quantity of information transmitted by a quantum
mechanical channel'',
\newblock \emph{Problems of Information Transmissions} \textbf{9}, pp.
177--183, 1973.

\bibitem{jozsaschlienzparadox}
Jozsa, R. and Schlienz, J.,
\newblock ``Distinguishability of states and the von Neumann entropy'',
\newblock \emph{Physical Reviews A} \textbf{62}, article 012301, 2000.

\bibitem{tradeoffinformationdisturbance}
Maccone, L.,
\newblock ``Information-disturbance tradeoff in quantum measurements'',
\emph{Physical Reviews A} \textbf{73}, article 042307, 2006.

\bibitem{book_machine_learning}
Mitchell, T.,
\newblock \emph{Machine Learning},
\newblock McGraw Hill,
1997.

\bibitem{discriminationrandomstates}
Montanaro, A.,
\newblock ``On the distinguishability of random quantum states'',
\newblock \emph{Communication in Mathematical Physics} \textbf{273}(3), pp.~619--636, 2006 .

\bibitem{lowerboundpgm}
Montanaro, A.,
\newblock ``A lower bound on the probability of error in quantum state discrimination'',
\newblock \emph{Proceedings of IEEE Information Theory Workshop}, 2008.

\bibitem{nielsen:book00}
Nielsen, M.\,A. and Chuang, I.\,L.,
\newblock \emph{Quantum Computation and Quantum Information},
\newblock Cambridge University Press, 2000.

\bibitem{jointmeasurement:91}
Peres, A. and Wootters, W.\,K.,
\newblock ``Optimal detection of quantum information'',
\newblock \emph{Physical Reviews Letters} \textbf{66}(9), pp.~1119--1122,
1991.

\bibitem{oneagainstall}
Rifkin, R. and Klatau, A.,
\newblock ``In defense of one-vs-all classification'',
\newblock \emph{Journal of Machine Learning Research} \textbf{5}, pp.~101--141, 2004.

\bibitem{templatematching:02}
Sasaki, M. and Carlini, A.,
\newblock ``Quantum learning and universal quantum matching machine'',
\newblock \emph{Physical Reviews A} \textbf{66}(2), article 022303, 2002.

\bibitem{book_statistical_learning}
Vapnik, V.,
\newblock \emph{The Nature of Statistical Learning Theory},
\newblock Springer,
1995.

\bibitem{nocloningtheorem}
Wootters, W.\,K. and \.{Z}urek, W.\,C.,
\newblock ``A single quantum cannot be cloned'',
\newblock \emph{Nature} \textbf{66}, pp.~802--803, 1982.

\bibitem{costingreduction}
Zadrozny, B., Langford, J. and Naoki, A.,
\newblock ``Cost-sensitive learning by cost-proportionate example weighting'',
\newblock \emph{Proceedings of the 3rd IEEE International Conference on Data Mining (ICDM'03)}, pp.~435--442, 2003.


\end{thebibliography}
\end{document}